\documentclass{article}

\usepackage{amsfonts}
\usepackage{amssymb}
\usepackage{graphicx}
\usepackage{amsmath}%
\usepackage{color}

\newtheorem{theorem}{Theorem}

\newtheorem{corollary}{Corollary}

\newtheorem{definition}{Definition}

\newtheorem{remark}{Remark}

\newenvironment{proof}[1][Proof]{\noindent\textbf{#1.} }{\ \rule{0.5em}{0.5em}}

\newcommand{\jinf}[0]{J^{\infty}}

\newcommand{\refeqn}[1]{(\ref{#1})}

\begin{document}

\title{Solvable structures for  evolution PDEs admitting differential constraints }

\author{Francesco De Vecchi\thanks{Dipartimento di Matematica, Universit\`a degli Studi di Milano, via Saldini 50, 20133 Milano (Italy); {\tt francesco.devecchi@unimi.it}},
Paola Morando\thanks{DISAA, Universit\`a degli Studi di Milano, via Celoria 2, 20133 Milano (Italy); {\tt paola.morando@unimi.it}} }

\maketitle



\begin{abstract}
Solvable structures are exploited in order to find families of explicit solutions  to evolution PDEs admitting suitable differential constraints.
The effectiveness of the method is verified on several explicit examples.
\end{abstract}

\section{Introduction}
In last decades a great interest has been devoted to symmetry
reduction methods for both ordinary and partial differential
equations and in recent times several Authors provided  different
kinds of generalizations of the classical results of Lie and Cartan.
This led to the development of new techniques which have given a
significant  improvement to the subject
(see, e.g.,\cite{BlumanKumei,CatalanoMorando,Fels,GaetaMorando,Olver2,PucciSaccomandi}).\\
In particular, the geometric approach based on  jet
bundles allows the description of  an  $r$-order ordinary differential equation (ODE) as a finite dimensional submanifold $\mathcal{E}$
of  a suitable jet space $J^{r}(M,\mathbb{R}^m)$ (\cite{KrasVin,Olver1, Stormark}).
In this setting,  the knowledge of a symmetry  for the ODE leads to reduce by one the dimension of the submanifold $\mathcal{E}$ and, under suitable hypotheses,   this  can be interpreted  as a reduction of the order of the ODE. Moreover, using  the
identification of $J^{r}(M,\mathbb{R}^m)$ with a subspace of the tangent bundle
$T(J^{r-1}(M,\mathbb{R}^m))$, an $r$-order
ODE can be described as a one-dimensional distribution of  vector fields on
$J^{r-1}(M,\mathbb{R}^m)$. If a solvable $r$-dimensional algebra of symmetries for this distribution  is known, the solution to the ODE can be obtained by quadratures. Solvable structures provide an extension of this classical result, significantly  enlarging the class of
vector fields which can be used to integrate by quadratures a given ODE and, more in general, an integrable distribution of vector fields (\cite{BarcoPrince, BasarabHorwath, HartlAthorne, SherringPrince}).
This approach can be extended to first order scalar partial differential equations (PDEs) as well as  to PDEs with  one-dimensional Cauchy characteristic space,  which are naturally described  by a single vector field on a suitable finite-dimensional jet space. In this case, the knowledge of a solvable structure allows the explicit determination of the solutions to the PDE by integrating a given system of closed one-forms (\cite{Barco, Barco1, BarcoPrince1}).\\
On the other hand, when we consider a  system of $m$ evolution  PDEs in two independent variables of the form
\begin{equation}\label{evoleq}
u^i_t=f^i(t,x,u^j, u^j_x, u^j_{xx}, \ldots ),
\end{equation}
we have to  attach to \eqref{evoleq} all its differential consequences and  the evolution PDE can  be described as an infinite dimensional submanifold $\mathcal{E} \subset \jinf(\mathbb{R}^2,\mathbb{R}^m)$ such that  $\mathcal{C} \subset T\mathcal{E}$ (here $\mathcal{C}$ denotes  the Cartan distribution i.e. the formally integrable distribution on $\jinf(\mathbb{R}^2,\mathbb{R}^m)$ generated by the total derivatives).
Therefore  the knowledge of a symmetry for the PDE does not lead
 to a reduction of the  dimension of the submanifold $\mathcal{E}$
but can be exploited by  looking for a special class of  solutions which are invariant under the symmetry. This is equivalent to look for the solutions to a new overdetermined system obtained by appending the invariance condition to the original system of PDEs. An interesting generalization of this reduction method is
provided by the \emph{differential constraints}  method, consisting in appending to \eqref{evoleq} an overdetermined systems of PDEs of the form $\mathcal{L}=\{L^1(x,t,u,u_{\sigma})=0,...,L^k(x,t,u,u_{\sigma})=0\}$ such that the system $\mathcal{L}$ admits a general finite dimensional solution and  is compatible with  \eqref{evoleq}. Many reduction methods, such  as (conditional) Lie-B\"{a}cklund and non classical symmetry reductions, direct method of Clarkson and Kruskal, Galaktionov's nonlinear separation method and  others can be seen as particular instances of differential constraints method (see \cite{BlumanCole,ClarksonKruskal,Galaktionov,Ji,KamranMilsonOlver,Kruglikov,LeviWinternitz,Olver2,PucciSaccomandi1,Zhdanov}).

\medskip
\noindent In this paper we  use solvable structures in order to obtain families of solutions to  systems of evolution PDEs of the form \eqref{evoleq} admitting suitable differential constraints.
In particular we associate with a differential constraint  a finite dimensional submanifold $\mathcal{H} \subset \mathcal{E}$ such that the Cartan distribution $\mathcal{C}$ is tangent to $\mathcal{H}$. Hence solutions to the PDEs \refeqn{evoleq} satisfying the differential constraints are integral manifolds of $\mathcal{C}_{\mathcal{H}}$  (the Cartan distribution $\mathcal{C}$ restricted to  $\mathcal{H}$) and the problem of finding particular solutions to the PDE  reduces to the problem of finding integral submanifolds of the integrable distribution $\mathcal{C}_{\mathcal{H}}$ on the finite dimensional manifold $\mathcal{H}$.
In this setting solvable structures can be successfully exploited to obtain families of explicit solutions to  \refeqn{evoleq}.

\noindent The paper is organized as follows: is Section \ref{prel}  we give a geometrical description of differential constraints method for  evolution PDEs,  in Section \ref{solvablestructure} we present solvable structures method for evolution PDEs admitting differential constraints and in Section
\ref{examples} we apply previous results to
several explicit examples.

\section{Differential constraints for evolution PDEs}\label{prel}

Let  $M$ be a $2$-dimensional manifold and $(x,t)$ be a global coordinate system on $M$.
The standard coordinate system for $J^k(M,\mathbb{R}^m)$ is  $x,t,u^i,u^i_x,u^i_t,u^i_{xx},...u^i_{\sigma}$, where $\sigma$ is a multi-index with $|\sigma|\leq k$ and  $u^i_{\sigma}$ represents the derivative with respect the variables $x,t$ the number of time given by the multi-index $\sigma=(h,k)$.
A coordinate system for $\jinf(M,\mathbb{R}^m)$ is given by $(x,t,u^i,u^i_x,...,u^i_{\sigma},...)$ without any restriction on the multi-index $\sigma$.
It is well known that the Cartan distribution $\mathcal{C}$ on $\jinf(M,\mathbb{R}^m)$ is the $2$-dimensional distribution  generated by the  vector fields
\begin{eqnarray*}
D_x&:=\partial_x+u^i_x\partial_{u^i}+u^i_{xx}\partial_{u^i_x}+u^i_{xt}\partial_{u^i_{t}}+...\\
D_t&:=\partial_t+u^i_t\partial_{u^i}+u^i_{xt}\partial_{u^i_x}+u^i_{tt}\partial_{u^i_{t}}+...
\end{eqnarray*}
With any system of  evolution   PDEs of the form
\begin{equation}\label{evoleq1}
F^i:=u^i_t-f^i(t,x,u^j,u^j_x,u^j_{xx},...)=0,
\end{equation}
where $f^i \in \mathcal{C}^{\infty}(J^k(M,\mathbb{R}^m))$ do not depend on the derivatives $u^i_t,u^i_{xt},...$, it is possible to associate
 a submanifold $\mathcal{E}$ of $\jinf(M,\mathbb{R}^m)$  such that the restriction $\mathcal{C}_{\mathcal{E}}$ of Cartan distribution to $\mathcal{E}$ satisfies $\mathcal{C}_{\mathcal{E}}\subset T\mathcal{E}$. This means that $\mathcal{E}$ is defined by the equations
\begin{equation}\label{evoleq2}
F^i=0, \quad D^l_x(D^r_t(F^i))=0 \qquad \forall l, r \in \mathbb{N}
\end{equation}
and we can consider the natural    coordinate system  $x,t,u^i,u^i_r $ on $\mathcal{E}$ (here $u^i_r$ denotes  the derivative of the function $u^i$  with respect to $x$ $r$ times).\\
One of the most useful methods for determining particular explicit solutions to a system of evolution PDEs of the form \eqref{evoleq1} is to reduce it to a system of ODEs. This can be done by enlarging the original system of PDEs appending  compatible additional equations (called differential constraints or side conditions, see \cite{Ji,KamranMilsonOlver,Kruglikov,Olver2,OlverRosenau}).  If we look at the system \eqref{evoleq1} as a submanifold $\mathcal{E}\subset \jinf(M, \mathbb{R}^m)$
the differential constraints method is equivalent to find a suitable finite-dimensional submanifold $\mathcal{H}$ of $\mathcal{E}$.
The particular form of  equation \eqref{evoleq1} and
the explicit expression of the generators $ \bar{D}_x, \bar{D}_t$ of  $\mathcal{C}_{\mathcal{E}}$
\begin{eqnarray*}
\bar{D}_x&=&\partial_x+u^i_x\partial_{u^i}+u^i_{xx}\partial_{u^i_x}+u^i_{xxx}\partial_{u^i_{xx}}+...\\
\bar{D}_t&=&\partial_t+f^i\partial_{u^i}+	\bar{D}_x(f^i) \partial_{u^i_x}+ \bar{D}^2_x(f^i) \partial_{u^i_{xx}}+...
\end{eqnarray*}
suggest to consider differential constraints
$L^i(x,t,u^j,u^j_{x}, u^j_{xx}, \ldots , u^j_{n} )=0$  of the form
\begin{equation}\label{evoleq3}
L^i=u^i_{n_i}-g^i(t,x,u^j,u^j_{x},...),
\end{equation}
where
$\partial_{u^i_{n}}(g^j)=0$ for any $n \geq n_i$ and to look for a submanifold $\mathcal{H}\subset \mathcal{E}$ defined by
\begin{equation}\label{eq_constraints}
L^i=0, \qquad \bar{D}^r_x(L^i)=0.
\end{equation}
In this way we have that $\bar{D}_x \in T\mathcal{H}$,  but usually $\bar{D}_t \not \in T\mathcal{H}$.\\
This corresponds to the fact that, choosing arbitrary functions $g^i$, the system given by the evolution equations \refeqn{evoleq1} and  the differential constraints \refeqn{evoleq3} is not compatible and so the set of solutions is empty.
Therefore $\bar{D}_t  \in T\mathcal{H}$ expresses the compatibility condition between the  evolution equations and the differential constraints associated with $\mathcal{H}$.
\begin{definition}
Let $\mathcal{E}$ be a submanifold of $\jinf(M,\mathbb{R}^m)$ defined by  equations \eqref{evoleq2}. A finite-dimensional submanifold $\mathcal{H}\subset \mathcal{E}$ defined by \eqref{eq_constraints} is a \emph{constraint submanifold} for $\mathcal{E}$ if
\begin{equation}\label{equation_invariant}
\bar{D}_t(L^i)|_{\mathcal{H}}=0,
\end{equation}
where the evaluation on $\mathcal{H}$ consists in replacing the expression of $u^i_{k_i}$, with $k_i \geq n_i$, in terms of $x,t,u^i,u^i_{h_i}$, with $h_i < n_i$, using equations \eqref{eq_constraints}.
\end{definition}
Since $\bar{D}_t$ and $\bar{D}_x$ commute on $\mathcal{E}$,  relations \refeqn{equation_invariant} hold if and only if
\begin{equation}\label{equation_ideal}
\bar{D}_t(\bar{D}^r_x(L^i))|_{ \mathcal{H}}=0
\end{equation}
and $\mathcal{H}\subset \mathcal{E}$ is a constraint submanifold for $\mathcal{E}$ if and only if   $\bar{D}_t \in T\mathcal{H}$. \\
\begin{remark}\label{rem_solutions}
If $\mathcal{H}\subset \mathcal{E}$ is a constraint submanifold for $\mathcal{E}$, the restriction $\mathcal{C}_{\mathcal{H}}$ of the Cartan distribution $\mathcal{C}$ to $\mathcal{H}$ is a completely integrable distribution and any maximal integral submanifold of $\mathcal{C}_{\mathcal{H}}$ corresponds to a common solution to \eqref{evoleq2} and \eqref{eq_constraints}.
\end{remark}
Unfortunately equations \refeqn{equation_invariant} are usually  non-linear PDEs for the functions $g^i$ and solving them is as difficult as solving  the initial PDE. In this paper we do not address this general problem but we exploit the knowledge of some particular constraint submanifolds and of suitable solvable structures for the restricted Cartan distribution in order to find families of explicit solutions.

\section{Solvable structures and integrability}\label{solvablestructure}

In this section  we  recall some basic definitions and facts about  solvable structures in our framework.   The reader is referred to \cite{BarcoPrince,BasarabHorwath,CatalanoMorando,HartlAthorne,SherringPrince}  for a complete and general discussion of the subject. These results  will be used il the next Section to compute  families of explicit solutions to evolution PDEs for which a finite-dimensional constraint submanifold is known.\\
It is well known that, given a $k$-dimensional involutive distribution $K$ on an $n$-dimensional manifold $N$, the knowledge of a  solvable $(n-k)$-dimensional algebra $\mathcal{G}$ of nontrivial symmetries for $K$
 guarantees that maximal integral submanifolds for $K$ can be found by quadratures. The notion of solvable structure provides a generalization of this classical integrability result, avoiding the use of rectification of vector fields  and allowing solutions to be represented in the original coordinates of the problem.

\begin{definition}
\label{Def_solv} Let $\mathcal{H}$ be an  $r$-dimensional constraint submanifold for a system $\mathcal{E}$ of evolution PDEs  and  $\mathcal{C}_{\mathcal{H}}= \langle \tilde {D}_x, \tilde{D}_t \rangle$ be the Cartan distribution restricted to $\mathcal{H}$.
The vector fields $\{X_1,X_2,\ldots,X_{r-2}\}$ are a \emph{solvable structure} for $\mathcal{C}_{\mathcal{H}}$ if and only if, $\forall h\leq r-2$, the vector field  $X_h$ is a nontrivial symmetry of $\mathcal{C}_{\mathcal{H}} \oplus \langle X_1,\ldots,X_{h-1} \rangle$.
\end{definition}

\begin{theorem}\label{TheoSolvStruct}
Let $\mathcal{H}$ be an  $r$-dimensional orientable constraint  submanifold of a system $\mathcal{E}$ of evolution PDEs   and $\Omega$ be a volume form on $\mathcal{H}$. If   $\{X_1,X_2,\ldots,X_{r-2}\}$  is a solvable structure for   $\mathcal{C}_{\mathcal{H}}$ such that $\mathcal{C}_{\mathcal{H}} \oplus \langle X_1,\ldots,X_{r-2} \rangle =T\mathcal{H}$,   then the one-forms
\begin{equation}
\displaystyle{\Omega_i:=\frac {X_1 \lrcorner\ldots \lrcorner\hat{X}_i\lrcorner\ldots \lrcorner X_{r-2}\lrcorner \tilde {D}_x\lrcorner \tilde {D}_t \lrcorner \Omega}{X_1 \lrcorner\ldots \lrcorner X_{r-2}\lrcorner \tilde {D}_x\lrcorner \tilde {D}_t \lrcorner \Omega}}, \;\; i=1,\ldots r-2
\end{equation}
satisfy
\begin{equation}
\begin{array}{l}
d\Omega_{r-2}=0\\
d\Omega_i=0 \qquad \mathrm{mod}(\Omega_{i+1}, \ldots, \Omega_{r-2})\\
\end{array}
\end{equation}
and it is possible to explicitly compute $r-2$ first integrals for the $2$-dimensional distribution $\mathcal{C}_{\mathcal{H}}$ on $\mathcal{H}$.
\end{theorem}
The interested reader is referred to the original papers \cite{BarcoPrince,BasarabHorwath,HartlAthorne,SherringPrince} for a proof of this theorem.
\begin{corollary}
Let $\mathcal{H}$ be an  $r$-dimensional orientable constraint  submanifold of a system $\mathcal{E}$ of evolution PDEs. If   $\{X_1,X_2,\ldots,X_{r-2}\}$  is a solvable structure for   $\mathcal{C}_{\mathcal{H}}$ such that $\mathcal{C}_{\mathcal{H}} \oplus \langle X_1,\ldots,X_{r-2} \rangle =T\mathcal{H}$, it is possible to explicitly compute a family of solutions to  $\mathcal{E}$ depending on $r-2$ parameters.
\end{corollary}
\begin{proof}
The proof is an immediate consequence of  Theorem \ref{TheoSolvStruct} and Remark \ref{rem_solutions}.
\end{proof}
\begin{remark}\label{rem_abeliano}
Theorem \ref{TheoSolvStruct} can be specialized to the particular  case of a  distribution $\mathcal{C}_{\mathcal{H}} $ admitting a nontrivial $(r-2)$-dimensional Abelian symmetry algebra $\mathcal{G}= \langle X_1,\ldots,X_{r-2} \rangle$. In this case the function
$$
M=\frac {1}{X_1\lrcorner X_2\lrcorner \ldots X_{r-2}\lrcorner  \tilde {D}_x \lrcorner \tilde {D}_t \lrcorner \Omega}
$$
provides an integrating factor for all the one-forms
$$\beta_i:= X_1 \lrcorner\ldots \lrcorner\hat{X}_i\lrcorner\ldots \lrcorner X_{r-2}\lrcorner \tilde {D}_x\lrcorner \tilde {D}_t \lrcorner \Omega
$$
so that each form $\Omega_i=M \beta_i $ can be separately integrated  by quadratures.
\end{remark}

\section{Examples}\label{examples}
In this Section we apply previous results to several  examples of evolution PDEs in order to show the effectiveness of the combined use of differential constraints and solvable structures  to  compute families of explicit solutions.

\subsection {Burgers' equation}
Consider the Burgers' equation
\begin{equation}\label{burgers1}
u_{t}=u_{xx}+u_{x}^2
\end{equation}
and  the distribution $\mathcal{C}_{\mathcal{E}}$
generated by
$$
\begin{array}{l}
\displaystyle { \bar D_x= \partial_x + u_x \partial_u+ u_{xx}\partial_{u_x}+ u_{xxx} \partial_{u_{xx}} + \ldots } \\
\displaystyle { \bar D_t= \partial_t +(u_{xx}+u_{x}^2)\partial_u +
\bar D_x(u_{xx}+u_{x}^2) \partial_{u_x} + \bar D^{(2)}_x(u_{xx}+u_{x}^2) \partial_{u_{xx}}+
\ldots}
\end{array}
$$
If we  consider the submanifold $\mathcal{H}$ of $\mathcal{E}$ given by
$$
\mathcal{H}:=\{ g=u_{xxx}+2u_xu_{xx}=0, \; \;  \bar D^{(k)}_x g=0 ; \; \;  k
\in \mathbb{N} \},
$$
it is easy to prove that $\mathcal{H}$ is a constraint submanifold for \eqref{burgers1}, being
$$
\bar D_t(g)|_{\mathcal{H}}= \bar D_x^{(3)}(u_{xx}+u_x^2)+2u_x \bar D^{(2)}_x(u_{xx}+u_x^2)+2u_{xx} \bar D_x(u_{xx}+u_x^2)|_{\mathcal{H}}=0.
$$
Therefore we can consider the restrictions $\tilde D_x$ and $\tilde D_t$ of the vector fields $\bar D_x$ and $\bar D_t$ to $\mathcal{H}$. In particular, choosing $(x,t,u,u_x,u_{xx})$ as  coordinates on  $\mathcal{H}$, the restricted vector fields are
$$
\begin{array}{l}
\displaystyle { \tilde D_x =  \partial_x + u_x \partial_u+ u_{xx}\partial_{u_x} -2u_x u_{xx} \partial_{u_{xx}}}\\
\displaystyle { \tilde D_t= \partial_t +(u_{xx}+u_{x}^2)\partial_u}.
\end{array}
$$
Since $\langle \tilde  D_x, \tilde D_t \rangle$ is an involutive
distribution on $\mathcal{H}$, Theorem  \ref{TheoSolvStruct}
ensures that the non trivial symmetry algebra generated by
$$
\begin{array}{lll}
X_1=\partial_x,   & X_2=\partial_u,  & X_3=2t\partial_t + x
\partial_x -u_x \partial_{u_x}-2u_{xx}
\partial_{u_{xx}}
\end{array}
$$
can be used to compute  solutions to (\ref{burgers1}). In
fact,  if we take the volume form $\Omega=  dx \wedge dt\
\wedge du \wedge du_x \wedge du_{xx}$ on $\mathcal{H}$,  the commutation relations
$$
[X_1,X_2]=0, \quad [X_1, X_3]=X_1, \quad [X_2, X_3]=0,
$$
 ensure that the function
$$
M=\frac {1}{X_1 \lrcorner X_2 \lrcorner X_3 \lrcorner  \tilde D_x \lrcorner \tilde D_t \lrcorner\Omega }=\frac {1}{2u_{xx}(u_{xx}+u_x^2)}
$$
provides an integrating factor for the one-forms
$$
\begin{array}{l}
\beta_2= X_1 \lrcorner X_3 \lrcorner \tilde D_x \lrcorner \tilde  D_t \lrcorner \Omega \\
\beta_3= X_1 \lrcorner X_2 \lrcorner \tilde D_x \lrcorner \tilde D_t \lrcorner \Omega,  \\
\end{array}
$$
while the one-form $\Omega_1= M \beta_1= M X_2 \lrcorner X_3 \lrcorner \tilde  D_x \lrcorner \tilde  D_t \lrcorner \Omega$ is closed modulo $\Omega_3=M\beta_3$.

\noindent In order to compute explicit solutions we rewrite  $\Omega_3$ as
$$
\begin{array}{l}
\displaystyle{\Omega_3=\frac {1}{2(u_{xx}+u_x^2)}du_{xx}+\frac {u_x}{(u_{xx}+u_x^2)} du_x}=d\left( \frac 12 \ln(u_{xx}+u_x^2)\right)=dF_3.\\
\end{array}
$$
On the level manifolds $F_3= \frac 12 \ln(u_{xx}+u_x^2)=c_3$ we have
\begin{equation}\label{burgers_c3}
u_{xx}=k^2_3-u_x^2
\end{equation}
with $k_3^2=e^{2c_3}$. Moreover, using (\ref{burgers_c3}), we find

$$
\begin{array}{l}
\displaystyle{\Omega_2=M\beta_2=
du-k^2_3dt- \frac {u_x}{k_3-u_x^2} du_x= d\left( u-k^2_3t+\frac 12 \ln (k^2_3-u_x^2)\right)=dF_2 }
\end{array}
$$
and, on the   level manifolds $F_2=c_2$, we have
\begin{equation}\label{burgers_c2}
u=k^2_3t-\frac 12\ln(k^2_3-u_x^2)+c_2.
\end{equation}
Finally, considering the restriction of the one-form $\Omega_1=M\beta_1$ to (\ref{burgers_c3}), we get
$$
\begin{array}{l}
\displaystyle{\Omega_1= -dx + \frac {1}{k^2_3-u_x^2} du_x= d\left( -x+\frac {1}{2k_3}\ln(\frac {k_3+u_x}{k_3-u_x})\right)=dF_1}\\
\end{array}
$$
and, on the level manifolds $F_1=c_1$, we find
$$
\displaystyle{u_x=k_3\left(\frac {e^{2k_3(x+c_1)}-1}{e^{2k_3(x+c_1)}+1}\right)}.
$$
Hence using  (\ref{burgers_c2}) we obtain the explicit solution to (\ref{burgers1}) in the form
$$
\displaystyle{u=k_3^2t-\ln\left( \frac {2k_3e^{k_3(x+c_1)}}{e^{2k_3(x+c_1)}+1}\right) +c_2.}
$$

\subsection{ Heat equation}

Let  consider  the heat  equation
\begin{equation}\label{heat1}
u_{t}=u_{xx}
\end{equation}
and  the corresponding distribution
generated by
$$
\begin{array}{l}
\displaystyle { \bar D_x= \partial_x + u_x \partial_u+ u_{xx}\partial_{u_x}+ u_{xxx} \partial_{u_{xx}} + \ldots } \\
\displaystyle { \bar D_t= \partial_t +u_{xx}\partial_u +
\bar D_x(u_{xx}) \partial_{u_x} +\bar D^{(2)}_x(u_{xx}) \partial_{u_{xx}}+
\ldots}
\end{array}
$$
The knowledge of a recursion operator for the heat equation provides  an infinite family of Lie-B\"{a}cklund symmetries of the form
$$
\begin{array}{l}
X_1=u\partial_u+ u_x\partial_{u_x}+ u_{xx}\partial_{u_{xx}}\ldots ,\\
X_2=u_x\partial_u+ u_{xx}\partial_{u_x}+ u_{xxx}\partial_{u_{xx}}\ldots ,\\
X_3=u_{xx}\partial_u+  u_{xxx}\partial_{u_x}+ u_{xxxx}\partial_{u_{xx}}\ldots ,\\
\ldots\\
X_n= u_{n-1}\partial_u+ u_n \partial_{u_x}+  u_{n+1} \partial_{u_{xx}}+ \ldots \\
\ldots\\
\end{array}
$$
In order to reduce to a finite-dimensional manifold,  we  consider the submanifolds $\mathcal{H}_n$ of $\mathcal{E}$ given by
$$
\mathcal{H}_n:=\{ g=u_{n}=0, \; \;  \bar D^{(k)}_x g=0 ; \; \;  k
\in \mathbb{N} \}.
$$
It is easy to prove that, $\forall n \in \mathbb{N}$, $\mathcal{H}_n$ is a constraint  submanifold (corresponding to the Lie-B\"{a}cklund symmetry $X_n$) for \eqref{heat1} so that we can consider the restrictions $\tilde D_x$ and $\tilde D_t$ of the vector fields $\bar D_x$ and $\bar D_t$ to $\mathcal{H}_n$. In particular, choosing $(x,t,u,u_x,u_{xx}, \ldots , u_{n-1})$ as  coordinates on  $\mathcal{H}_n$, the restricted vector fields are
$$
\begin{array}{l}
\displaystyle { \tilde D_x =  \partial_x + u_x \partial_u+ u_{xx}\partial_{u_x} +u_{xxx}\partial_{u_{xx}}+ \ldots +u_{n-1}\partial_{u_{n-2}}}\\
\displaystyle { \tilde D_t= \partial_t +u_{xx}\partial_u+ u_{xxx}\partial_{u_x}+ u_{xxxx}\partial_{u_{xx}}+ \ldots +u_{n-1}\partial_{u_{n-3}}}.
\end{array}
$$
Since $X_1, X_2, \ldots X_{n-1}$  are symmetries of $X_n$ and generate a non trivial Abelian symmetry algebra for $\langle \tilde  D_x, \tilde D_t \rangle$ (when $u_{n-1}\not= 0$), Remark \ref{rem_abeliano}  ensures   that they can be used to compute  explicit solutions to (\ref{heat1}) for any  $n\in \mathbb{N}$.

This example suggests  a wide  range of possible applications and  developments of the proposed method. Indeed a similar  procedure can be used whenever a local recursion operator of order one is known, allowing the construction of an infinite  family of commuting  Lie-B\"{a}cklund symmetries $X_i$ (with $i \in \mathbb{N}$). In fact,  fixing an order $n$ and considering the submanifold $\mathcal{H}_n$ corresponding to the vanishing of the generator of $X_n$ and its differential consequences,  all the vector fields $X_h$ with $h<n$  are  tangent to $\mathcal{H}_n$ and provide (on a suitable submanifold of $\mathcal{H}_n$ where they are independent) an Abelian symmetry algebra  of suitable dimension for the  distribution generated by the restricted vector fields $\tilde D_x$ and $\tilde D_t$.

\subsection{Modified heat equation}\label{mod_heat}

Consider the equation
\begin{equation}\label{modheat1}
u_t=au_{xx}+(bx+c)
\end{equation}
and the  distribution generated  by $\bar D_x$ and
$$\bar D_t=\partial_t+ \bar D_x(au_{xx}+(bx+c)u)\partial_{u_x}+ \bar D^{(2)}_x(au_{xx}+(bx+c)u)\partial_{u_{xx}}+\ldots$$
If we take
$$g=u_{xxx}-\frac{3u_xu_{xx}}{u}+\frac{2u_x^3}{u^2},$$
it is easy to verify that the finite-dimensional  submanifold $\mathcal{H}$ of $\mathcal{E}$ defined by
$$\mathcal{H}=\{g=0, \; \bar D^{(k)}_x(g)=0; \;\; k\in \mathbb{N}\}$$
is a constraint submanifold for \eqref{modheat1}.
 Using  $(t,x,u,u_x,u_{xx})$ as coordinates  on $\mathcal{H}$ and denoting  by $\tilde{D}_x,\tilde{D}_t$ the restriction of $\bar D_x,\bar D_t$ to $\mathcal{H}$, we have
\begin{eqnarray*}
\tilde{D}_x&=&\partial_x+u_x\partial_u+u_{xx}\partial_{u_x}+\left(\frac{3u_xu_{xx}}{u}-\frac{2u_x^2}{u^2}\right)\partial_{u_{xx}}\\
\tilde{D}_t&=&\partial_t+(au_{xx}+(bx+c)u)\partial_u+\\
&&\frac{3auu_xu_{xx}-2au^3_x+bu^3+bu_xu^2x+cu^2u_x}{u^2}\partial_{u_x}+\\
&&\frac{3au^2u^2_{xx}-2au^4_x+2bu^3u_x+bu^3u_{xx}x+cu^3u_{xx}}{u^3}\partial_{u_{xx}}.
\end{eqnarray*}
The three vector fields on $\mathcal{H}$
\begin{eqnarray*}
X_1&=&u\partial_u+u_x\partial_{u_x}+u_{xx}\partial_{u_{xx}}\\
X_2&=&\partial_x\\
X_3&=&\partial_t
\end{eqnarray*}
form a solvable structure for $\tilde D_x,\tilde D_t$,  being
\begin{eqnarray*}
&[X_1,X_2]=[X_1,X_3]=[X_2,X_3]=0&\\
&[X_1,\tilde D_x]=[X_2,\tilde D_x]=[X_3,\tilde D_x]=0&\\
&[X_1,\tilde D_t]=[X_3,\tilde D_t]=0&\\
&[X_2,\tilde D_t]=bX_1.&
\end{eqnarray*}
Since   $X_1,X_2,X_3,\tilde{D}_x,\tilde{D}_t$ are linearly independent on $\mathcal{H}$, we can consider the volume form $\Omega=dt \wedge dx \wedge du \wedge du_x \wedge du_{xx}$ on $\mathcal{H}$ and the function $M$ defined by
\begin{eqnarray*}
\frac 1 M &=&X_1 \lrcorner X_2 \lrcorner X_3 \lrcorner \tilde{D}_x \lrcorner \tilde{D}_t \lrcorner \Omega\\
&=&2a\frac{-u^3u^3_{xx}+3u^2u^2_xu_{xx}^2-3uu^4_xu_{xx}+u^6_x}{u^3}.
\end{eqnarray*}
In order to simplify  notations we use the new coordinate
$$v=\frac{u_{xx}}{u}-\frac{u^2_x}{u^2}$$
instead of $u_{xx}$ and we obtain
$$M=-\frac{1}{2av^3u^3}.$$
Since $X_1,X_3,\tilde{D}_x,\tilde{D}_t$ commute, $M$ is and integrating factor for $\beta_2= X_1 \lrcorner X_3 \lrcorner \tilde{D}_x \lrcorner \tilde{D}_t \lrcorner \Omega$ and $\beta_3= X_1 \lrcorner X_2 \lrcorner \tilde{D}_x \lrcorner \tilde{D}_t \lrcorner \Omega$  and
$$
\begin{array}{l}
\displaystyle{\Omega_2 = M \beta_2=
dx+\frac{u_x}{vu^2}du-\frac{1}{vu}du_x+\frac{2avu_x+bu}{2av^3u}dv}\\
\displaystyle{\Omega_3=M \beta_3= dt-\frac{1}{2av^2}dv}\\
\end{array}
$$
are closed differential forms. The integral functions of $\Omega_2,\Omega_3$ are
\begin{eqnarray*}
F_2&=&x+\frac{4avu_x-bu}{4av^2u}\\
F_3&=&t-\frac{1}{2av},
\end{eqnarray*}
and, on the level manifolds $F_2=k_2, F_3=k_3$, we get
\begin{eqnarray*}
v&=&\frac{1}{2a(k_2-t)}\\
u_x&=&\frac{u\left(-abk_2^2+2abk_2t-abt^2-k_3+x\right)}{2a\left(k_2-t\right)}.
\end{eqnarray*}
Moreover, from Theorem \ref{TheoSolvStruct} we have that
$$\Omega_1=M X_2 \lrcorner X_3 \lrcorner \tilde{D}_x \lrcorner \tilde{D}_t \lrcorner \Omega$$
is closed modulo $dF_2,dF_3$ and
integrating $\Omega_1$ we find
\begin{equation}\label{equation_heat1}
\begin{array}{lcl}
F_1&=&-\frac{1}{2}\log(t-k_2)-\log(u)+\left(3a^2b^2k_2^3t-6a^2b^2k_2^2t^2+\right.4a^2b^2k_2t^3\\
&&\\
&&-a^2b^2t^4-6abk_2^2x+6abk_2k_3t+12abk_2tx-6abk_3t^2-\\
&&\\
&&6abt^2x+12ack_2t-12act^2-6k_3x+3x^2+\\
&&\\
&&\left.3k_3^2\right)/\left(12a\left(k_2-t\right)\right).
\end{array}
\end{equation}
We can solve equation $F_1=k_1$  with respect to $u$ in order to write the explicit solution
\begin{eqnarray*}
u&=&\frac{1}{\sqrt{(t-k_2)}}\exp\left(\frac{x^2}{4a(k_2-t)}+\frac{(-k_3-abt^2+2abk_2t-abk_2^2)x}{2a\left(k_2-t\right)}+\right.\\
&&\frac{-a^2b^2t^4+4a^2b^2k_2t^3-(6a^2b^2k_2^2+6abk_3+12ac)t^2}{12a\left(k_2-t\right)}+\\
&&\left.\frac{(a^2b^2k_2^3+2abk_2k_3+4ack_2)t+k_3^2}{4a\left(k_2-t\right)}-k_1\right).
\end{eqnarray*}
If, for example, we chose $a=1,b=0,c=0,k_1=\ln(\sqrt{4\pi}), k_2=0,k_3=y$ we obtain
$$u=\frac{1}{\sqrt{4\pi t}}\exp\left(\frac{-(x-y)^2}{4t}\right)$$
that is the well know heat kernel.

\subsection{System of evolution equations}

Consider the following system of evolution equations

\begin{equation}\label{sist_evol}
u_{t}=u_{xx}+\frac 12 v^2 , \qquad
v_t=2v_{xx}
\end{equation}
and the corresponding  restricted Cartan distribution
generated by
$$
\begin{array}{l}
\displaystyle { \bar D_x=  \partial_x + u_x \partial_u+v_x \partial_v+ u_{xx} \partial_{u_x}+v_{xx} \partial_{v_x} + \ldots }\\
\displaystyle { \bar D_t= \partial_t+(u_{xx}+\frac 12v^2)\partial_u +2v_{xx} \partial_v+
\bar D_x(u_{xx}+\frac 12v^2) \partial_{u_x}+\bar D_x(2v_{xx}) \partial_{v_x}  +
\ldots}
\end{array}
$$
If we  consider the submanifold $\mathcal{H}$ of $\mathcal{E}$ given by
$$
\mathcal{H}:=\{ g^1=u_{xxx}+3v v_x=0, \; \; g^2=v_{xxx}=0, \; \;  \bar D^{(k)}_x g^1=0 ,  \; \; \bar D^{(k)}_x g^2=0; \; \;  k
\in \mathbb{N} \}
$$
it is easy to prove (by explicit computation) that $\mathcal{H}$ is a constraint submanifold for \eqref{sist_evol}. Hence
 we can restrict  the vector fields $\bar D_x$ and $\bar D_t$ to $\mathcal{H}$ and, choosing  coordinates $(x,t,u,v,u_x,v_x,u_{xx}, v_{xx})$ on  $\mathcal{H}$,  we find
$$
\begin{array}{l}
\displaystyle {\tilde D_x =  \partial_x + u_x \partial_u+ v_x \partial_v+
 u_{xx} \partial_{u_x}+v_{xx} \partial_{v_x}-3v v_x \partial_{u_{xx}}}\\
\displaystyle { \tilde D_t= \partial_t+(u_{xx}+\frac 12 v^2)\partial_u+ 2v_{xx} \partial_v-2vv_x \partial_{u_x}-(2v_x^2+2vv_{xx}) \partial_{u_{xx}}}.
\end{array}
$$
Since $\langle \tilde D_x, \tilde D_t \rangle$ is an involutive
distribution on $\mathcal{H}$, we can use  the solvable structure
$$
\begin{array}{lll}
X_1=\partial_t,  & X_2= \partial_x,  & X_3=\partial_u,\\
X_4= \partial_{u_x}, & X_5= \partial_{u_{xx}} , &  X_6=\partial_{v_{xx}}, \\
\end{array}
$$
in order to find explicit solutions to the system (\ref{sist_evol}).
If we take  the volume form on $\mathcal{H}$ as
$\Omega=dt \wedge dx \wedge du \wedge du_x \wedge du_{xx} \wedge dv \wedge dv_x \wedge dv_{xx}$
 we find
$$
\frac 1M= X_1 \lrcorner X_2 \lrcorner X_3 \lrcorner X_4 \lrcorner X_5 \lrcorner X_6 \lrcorner \tilde D_x \lrcorner \tilde D_t \lrcorner \Omega\\
= - 2 v_{xx}^2
$$
and we get the closed one-form
$$\Omega_6=M  X_1 \lrcorner X_2 \lrcorner X_3 \lrcorner X_4 \lrcorner X_5 \lrcorner \tilde  D_x \lrcorner \tilde D_t \lrcorner \Omega=dv_{xx}=dF_6.$$
Moreover, on the level manifolds $F_6=c_6$, the one-form
\begin{eqnarray*}
\Omega_5&=&M  X_1 \lrcorner X_2 \lrcorner X_3 \lrcorner X_4 \lrcorner X_6 \lrcorner \tilde D_x \lrcorner \tilde D_t \lrcorner \Omega\\
&=&\frac{c_{6}^2du_{xx}+(c_{6}^2v+c_{6}v_x^2)dv+(2c_{6}vv_x-v_x^3)dv_x}{c_{6}^2} \\
\end{eqnarray*}
is closed and
$$\Omega_5=dF_5=d\left(u_{xx}-\frac{v_x^4}{4c_6^2}+\frac{v^2}{2}+\frac{vv_x^2}{c_6}\right).$$
Furthermore, if we restrict to $F_5=c_5, \, F_6=c_6$, the one-form
\begin{eqnarray*}
\Omega_4&=&M  X_1 \lrcorner X_2 \lrcorner X_3 \lrcorner X_5 \lrcorner X_6 \lrcorner \tilde  D_x \lrcorner \tilde D_t \lrcorner \Omega\\
&=&\frac{-4c_6^3du_x-4c_{6}^2vv_xdv+(4c_{6}^2c_{5}-2c_{6}^2v^2+v_x^4)dv_x}{4c_{6}^3}
\end{eqnarray*}
satisfies
$$\Omega_4=dF_4=d\left( -u_x+\frac{c_5v_x}{c_6}-\frac{v^2v_x}{2c_6}+\frac{v_x^5}{20c_6^3}\right).$$
Then, on the level manifolds $F_4=c_4, F_5=c_5,F_6=c_6$ we find
\begin{eqnarray*}
\Omega_3 & = &M  X_1 \lrcorner X_2 \lrcorner X_4 \lrcorner X_5 \lrcorner X_6 \lrcorner \tilde D_x \lrcorner \tilde  D_t \lrcorner \Omega\\
&=&\left( 40c_{6}^4du+(-20c_{6}^3c_{5}+20c_{6}^2vv_x^2-5c_{6}v_x^4)dv+\right.\\
&&\left. (-40c_{6}^3c_{4}-20c_{6}^2c_{5}v_x+20c_{6}^2v^2v_x-20c_{6}vv_x^3+3v_x^5)dv_x \right)/ \left(40c_{6}^4\right)\\
\end{eqnarray*}
and
$$\Omega_3=dF_3=d\left(u+\frac{-40c_{6}^3c_{5}v-80c_{6}^3c_{4}v_x-20c_{6}^2c_{5}v_x^2+20c_{6}^2v^2v_x^2-10c_{6}vv_x^4+v_x^6}{80c_{6}^4}\right),$$
while, on the submanifolds $F_6=c_6$, we have
\begin{eqnarray*}
\Omega_2&=&M  X_1 \lrcorner X_3 \lrcorner X_4 \lrcorner X_5 \lrcorner X_6 \lrcorner \tilde D_x \lrcorner \tilde  D_t \lrcorner \Omega\\
&=&\frac{-dv_x+c_{6}dx}{c_{6}},
\end{eqnarray*}
so that
$$\Omega_2=dF_2=d(v_x-c_6 x).$$
Finally,  on the level manifolds $F_2=c_2,F_6=c_6$ we have
\begin{eqnarray*}
\Omega_1&=&M X_2 \lrcorner X_3 \lrcorner X_4 \lrcorner X_5 \lrcorner X_6 \lrcorner \tilde  D_x \lrcorner \tilde D_t \lrcorner \Omega\\
&=&\frac{-2c_{6}dt+dv+(-c_{6}x-c_{2})dx}{2c_{6}}\\
&=&dF_1=d\left( -v-2c_6\left(-\frac{x^2}{4}-\frac{c_2x}{2c_6}-t\right) \right)
\end{eqnarray*}
and from $F_1=c_1$ and  the previous equations we obtain the explicit solution
\begin{eqnarray*}
v&=&-c_1-2c_6\left(-\frac{x^2}{4}-\frac{c_2x}{2c_6}-t\right)\\
u&=&\left(-80c_{6}^6t^2x^2-20c_{6}^6tx^4-c_{6}^6x^6-160c_{6}^5c_{2}t^2x-80c_{6}^5c_{2}tx^3-6c_{6}^5c_{2}x^5+\right.\\
&&80c_{6}^5c_{1}tx^2+10c_{6}^5c_{1}x^4+80c_{6}^4c_{5}t+40c_{6}^4c_{5}x^2+80c_{6}^4c_{4}x-80c_{6}^4c_{2}^2t^2+\\
&&-80c_{6}^4c_{2}^2tx^2-10c_{6}^4c_{2}^2x^4+160c_{6}^4c_{2}c_{1}tx+40c_{6}^4c_{2}c_{1}x^3-20c_{6}^4c_{1}^2x^2+\\
&&-80c_{6}^4c_{3}+80c_{6}^3c_{5}c_{2}x-40c_{6}^3c_{5}c_{1}+80c_{6}^3c_{4}c_{2}+80c_{6}^3c_{2}^2c_{1}t+\\
&&40c_{6}^3c_{2}^2c_{1}x^2-40c_{6}^3c_{2}c_{1}^2x+20c_{6}^2c_{5}c_{2}^2+20c_{6}^2c_{2}^4t+10c_{6}^2c_{2}^4x^2+
\\
&&\left.-20c_{6}^2c_{2}^2c_{1}^2+4c_{6}\*c_{2}^5x-10c_{6}c_{2}^4c_{1}-c_{2}^6\right)/\left(80c_{6}^4\right).
\end{eqnarray*}

\end{document}